\documentclass[conference,letterpaper]{IEEEtran}

\usepackage[utf8]{inputenc} 
\usepackage[T1]{fontenc}
\usepackage{url}
\usepackage{ifthen}
\usepackage{cite}
\usepackage[cmex10]{amsmath} 

\usepackage{xcolor}
\usepackage{graphicx}
\usepackage{physics}
\usepackage{amssymb}
\usepackage{comment}
\usepackage{algorithm}
\usepackage{algpseudocode}
\usepackage[colorlinks,pdfstartview=FitH,citecolor=blue,linkcolor=blue,urlcolor=blue]{hyperref}
\usepackage{amsthm}
\newtheorem{lemma}{Lemma}
\newtheorem{theorem}{Theorem}
\newcommand{\Gr}{\text{Gr}}
\newcommand{\CS}{\text{CS}}
\newcommand{\Id}{\text{Id}}
\newcommand{\SWAP}{\text{SWAP}}

\newcommand{\F}{\mathbb{F}}
\newcommand{\C}{\mathbb{C}}
\newcommand{\E}{\mathbb{E}}
\newcommand{\mcH}{\mathcal{H}}
\newcommand{\mcX}{\mathcal{X}}

\makeatletter
\renewcommand{\fnum@algorithm}{} 
\makeatother

\begin{document}
\title{Winning Rates of $(n,k)$ Coset Monogamy Games} 

\author{%
  \IEEEauthorblockN{Michael Schleppy and Emina Soljanin}
  \IEEEauthorblockA{
  Rutgers University\\
                    Email: \{michael.schleppy, emina.soljanin\}@rutgers.edu
                    }
    \thanks{M.~Schleppy and E.~Soljanin are with the Department of Electrical and Computer Engineering, Rutgers, the State University of New Jersey, Piscataway, NJ 08854, USA, e-mail: \{michael.schleppy, emina.soljanin\}@rutgers.edu).}
    \thanks{This  research  is in part based  upon  work  supported  by  the  National  Science  Foundation  under  Grant  \# FET-2007203}
}


\IEEEoverridecommandlockouts

\maketitle


\begin{abstract}
    We formulate the $(n,k)$ Coset Monogamy Game, in which two players must extract complementary information of unequal size ($k$ bits vs.\ $n-k$ bits) from a random coset state without communicating. The complementary information takes the form of random Pauli-X and Pauli-Z errors on subspace states. Our game generalizes those considered in previous works that deal with the case of equal information size $(k=n/2)$. We prove a convex upper bound of the information-theoretic winning rate of the $(n,k)$ Coset Monogamy Game in terms of the subspace rate $R=\frac{k}{n}\in [0,1]$. This bound improves upon previous results for the case of $R=1/2$, in part due to a structural result we prove on subspace permutations, which may have broader combinatorial interest. We also prove the achievability of an optimal winning probability upper bound for the class of unentangled strategies of the $(n,k)$ Coset Monogamy Game.
\end{abstract}


\section{Introduction}
A valuable concept in Quantum Security is using a random basis to encode classical information into a quantum state. This feature plays an indispensable role in cryptographic schemes such as the BB84 Quantum Key Distribution (QKD) \cite{QKD:bennett2014quantum}, or Quantum Money \cite{QMoney:wiesner1983conjugate}: the classical data encoded in a quantum state $\ket{\psi}$ remains 'hidden' to anyone in possession of $\ket{\psi}$ without knowledge of the basis in which it was prepared (the encoding basis). In many quantum cryptographic protocols, the encoding basis is revealed once the transmitted quantum state is in the possession of a trusted receiver, allowing the classical information to be extracted from the quantum state. A natural question arises about whether arbitrary quantum states can be replicated so that separated parties can simultaneously, without communicating, extract the same classical data after the encoding basis is announced. The no-cloning theorem \cite{NoCloning:wootters1982single, NoCloning:dieks1982communication} asserts that such replication is generally impossible without knowledge of the encoding basis. This impossibility prevents, for instance, undetectable extraction of secret key bits in a QKD protocol or illegal duplication of a Quantum Banknote in a Quantum Money scheme.

In this work, we focus on encoding bases consisting of coset states. The coset states in \eqref{eq:CosetState} are a natural generalization of the BB84 states $\{\ket{0}, \ket{1}, \ket{+}, \ket{-}\}$ over $n$-qubits, defined by a subspace $W \subseteq \F_2^n$ and bit strings $x,z \in \F_2^n$. Each subspace $W$ determines a unique basis of coset states over $n$-qubits, whereas $(x,z)$ can together encode up to $n$ classical bits of information (see Section \ref{sec:prelim}). A useful feature of coset states is the ability to extract partial information about $x$ or $z$ (but not both) without knowing the basis. This is possible by either measuring in the computational basis on every qubit or measuring in the Hadamard basis on every qubit, respectively. However, measuring coset states in the computational basis destroys the encoded information of $z$, and measuring coset states in the Hadamard basis destroys the encoded information of $x$. Nevertheless, this property is desirable in various one-time Quantum Money protocols \cite{QMoney:ben2023quantum,QMoney:vidick2021classical,QMoney:coladangelo2021hidden}. In contrast to BB84 states over $n$-qubits, coset states may exhibit a degree of entanglement between qubits, allowing the information bits of $x$ and $z$ to be 'spread out' over multiple qubits instead of being localized at a single qubit. The encoding circuits for coset states were discussed in \cite{Game:schleppy2024optimal}.

Coset states have been investigated for their utility in cryptographic protocols. In \cite{QMoney:vidick2021classical}, the authors show that an adapted version of the Quantum Money scheme of \cite{QMoney:aaronson2012quantum} using coset states (referred to as one-time padded subspace states) constitutes a Proof of Quantum Knowledge protocol. The authors of \cite{QMoney:coladangelo2021hidden} establish a signature token scheme using a computational direct hardness property for coset states. They also demonstrate an unclonable decryption scheme using indistinguishability obfuscation and post-quantum one-way functions, conditional on an information-theoretic conjectured monogamy property of coset states. This conjectured monogamy property can be formulated in terms of a coset guessing game, in which Alice prepares a random coset state using a random subspace $W \subseteq \F_2^n$ of dimension $\frac{n}{2}$, before sending the state over a fixed quantum channel $\Phi$ to Bob and Charlie and announcing $W$ publicly. In \cite{Game:culf2022monogamy}, the conjectured monogamy property was proven by showing that the probability that Bob and Charlie can simultaneously determine $x$ and $z$ respectively up to a coset of $W$ and the dual space $W^\perp$ decays exponentially as $n$ increases. In \cite{Game:schleppy2024optimal}, optimal measurements in the coset guessing game were found for channels that split the qubits equally between Bob and Charlie. 

In the present work, we generalize the quantum coset guessing game of \cite{Game:culf2022monogamy,Game:schleppy2024optimal} to a more general setting called the $(n,k)$ Coset Monogamy Game, where the coset states are prepared in bases of subspaces $W \subseteq \F_2^n$ of a fixed dimension $0 \leq k \leq n$. This introduces an inequality in the number of bits of information to be guessed by Bob and Charlie, so that Bob needs to guess $n-k$ bits of information, and Charlie needs to guess $k$ bits of information. Previous works \cite{Game:culf2022monogamy,Game:schleppy2024optimal} have considered solely the case where $k=\frac{n}{2}$. The winning rate $\omega(R)$ for $R \in [0,1]$ is particularly interesting. It represents an information-theoretic quantity that relates the probability of a simultaneous correct guess per qubit to $n$ going to infinity. Phrased in terms of the subspace rate $R \in [0,1]$, Bob needs to guess approximately $n(1-R)$ bits of information, and Charlie needs to guess approximately $nR$ bits of information. We prove a convex upper bound for the winning rate $\omega(R) \leq 2^{-\frac{1}{2}R^*}$, where $R^* = \min \{R,1-R\}$.

The paper is organized as follows: In Section \ref{sec:prelim}, we introduce some basic preliminaries, notation, and some technical lemmas. To our knowledge, Lemma \ref{lem:subspacePermutations} appears to be novel, and we believe that the result may be of independent interest for subspace coding and designs. In Section \ref{sec:def}, we define the $(n,k)$ Coset Monogamy Game. In Section \ref{sec:thm}, we prove the main theorem regarding the winning rates. Finally, in Section \ref{sec:other}, we prove an optimal upper bound for a class of unentangled strategies for the extended non-local formulation of the $(n,k)$ Coset Monogamy Game.


\section{Preliminaries and Notation}\label{sec:prelim}
Throughout the paper, $k,n$ are integers with $0 \leq k \leq n$. The Grassmannian $\Gr_2(n,k)$ denotes the set of $k$-dimensional subspaces of the $n$-dimensional vector space $\F_2^n$. The rate of a subspace $W\in \Gr_2(n,k)$ is given by $R=\frac{k}{n}$. The Gaussian binomial coefficient $\binom{n}{k}_2 = \frac{\prod_{i=0}^{k-1}2^{n-i} - 1}{\prod_{i=0}^{k-1}2^{k-i} - 1}$ gives the cardinality of $\Gr_2(n,k)$. For a subspace $W \in \Gr_2(n,k)$ and a vector $x \in \F_2^n$, the coset of $W$ containing $x$ is denoted by $x+W$. We choose one representative from each coset of $W$ and identify each coset by its representative. The set $\CS(W)$ will denote the set of coset representatives for $W$. 

We will use the standard bilinear form on $\F_2^n$ defined by $x \cdot y = \sum_{i=1}^n x_iy_i$. For a subspace $W \in \Gr_2(n,k)$, the dual space is defined by $W^\perp = \{y \in \F_2^n : x \cdot y = 0 \quad \forall x \in W\}$, with $\dim(W^\perp) = n-k$. A vector $z \in \F_2^n$ defines a character $\chi_z : \F_2^n \to \F_2$ by the mapping $\chi_z(x) = z \cdot x$. When restricted to a subspace $W$, two characters $\chi_z, \chi_{z'}$ are equivalent iff $z-z' \in W^\perp$ (i.e. they belong to the same coset of $W^{\perp}$). We, therefore, identify characters over $W$ by the set of coset representatives of $W^\perp$, namely $\CS(W^\perp)$.

Given a subspace $W \in \Gr_2(n,k)$, we define the subspace state $\ket{W} \in (\C^2)^{\otimes n}$ over $n$ qubits by
\begin{equation}\label{eq:SubspaceState}
    \ket{W} = \frac{1}{\sqrt{2^k}} \sum_{u \in W} \ket{u}
\end{equation}
Additionally, given vectors $x,z \in \F_2^n$, we define the coset state $\ket{W_{x,z}} \in (\C^2)^{\otimes n}$ over $n$ qubits by
\begin{equation}\label{eq:CosetState}
    \ket{W_{x,z}} = \frac{1}{\sqrt{2^k}} \sum_{u \in W}(-1)^{z \cdot u} \ket{x+u} = X^xZ^z\ket{W}
\end{equation}
where $X^x= \bigotimes_{i=1}^n X^{x_i}$ and $Z^z= \bigotimes_{i=1}^n Z^{z_i}$. In this view, coset states are subspace states with Pauli-X and Pauli-Z errors applied to the $n$ qubits according to $x$ and $z$, respectively. Two coset states over the same subspace $\ket{W_{x,z}}, \ket{W_{x',z'}}$ are equivalent up to a global phase when $x-x' \in W$ and $z-z' \in W^{\perp}$, and are orthogonal otherwise. We may therefore associate to each subspace $W$ a basis $\{ \ket{W_{x,z}}\}_{x,z}$ of $(\C^2)^{\otimes n}$ for $x \in \CS(W)$ and $z \in \CS(W^\perp)$, where $x$ encodes $n-k$ bits of information and $z$ encodes $k$ bits of information. The norm of the inner product of two coset states, in general, is given by
\begin{equation}\label{eq:innerProd}
    |\bra{V_{x,z}} \ket{W_{x',z'}}| = 
    \begin{cases}
        2^{\dim(V \cap W) - k} & \text{if } \substack{x-x' \in V+W \\ z-z' \in V^\perp + W^\perp}\\
        0 & \text{otherwise}
    \end{cases}
\end{equation}
Applying the $n$-fold Hadamard operator $H^{\otimes n}$ to a coset state, we have the dual relation $H^{\otimes n} \ket{W_{x,z}} = \ket{W_{z,x}^\perp}$.

A quantum state on a finite-dimensional Hilbert space $\mcH$ can be described by a density matrix $\rho$ (a positive semidefinite matrix on $\mcH$ with unit trace). A rank-one density matrix $\ket{\psi}\bra{\psi}$ is termed a pure state. A quantum channel $\Phi: \mcH_1 \to \mcH_2$ defines an operation on a quantum state, which is a completely positive and trace-preserving linear map. A Positive Operator-Valued Measure (POVM) on a Hilbert space $\mcH$ is a set of positive semidefinite matrices $\{P_x\}_{x \in \mcX}$ defined on $\mcH$ satisfying $\sum_{x \in \mcX} P_x = \Id_\mcH$. In the case that each $\{P_x\}_{x \in \mcX}$ are orthogonal projections, that is, $P_xP_{x'} = \delta_{xx'}P_x$, we say that $\{P_x\}_{x \in \mcX}$ is a Projective-Value Measure (PVM). If $\rho$ is a quantum state in $\mcH$, then the probability of obtaining the result $x \in \mcX$ when measuring $\rho$ with the POVM $\{P_x\}_{x \in \mcX}$ is $p_x = \Tr[P_x\rho]$. For operators $P$ and $Q$, we write $0\leq P$ when $P$ is positive semidefinite and $P \leq Q$ when $0 \leq Q-P$.

Finally, we list four lemmas that we will use to prove the main result of this paper. The omitted proofs can be found in the cited references.

\begin{lemma}\label{lem:NormOfSum}
    \cite[Lemma 2]{Game:tomamichel2013monogamy} Let $\{P_i\}_{i=1}^n$ be a set of positive semidefinite operators on a finite-dimensional Hilbert space $\mcH$. For any set of mutually orthogonal permutations $\{\pi_i\}_{i=1}^n$ (that is, $\pi_i(k) \neq \pi_j(k)$ for $i \neq j$), we have
    \begin{equation}
        \| \sum_{i=1}^n P_i \| \leq \sum_{j = 1}^n \max_{1 \leq i \leq n} \|\smash[b]{\sqrt{P_i}} \sqrt{\smash[b]{P_{\pi_j(i)}}}\|
    \end{equation}
\end{lemma}
We can omit the square roots when $\{P_i\}_{i=1}^n$ are projections. The next lemma extends the result of \cite[Lemma 3.2]{Game:culf2022monogamy} to arbitrary $0 \leq k \leq n$:

\begin{lemma}\label{lem:prodOfSumOfCosetStates}
    For $V,W \in \Gr_2(n,k)$, $z \in \CS(V^{\perp})$ and $x' \in \CS(W)$ we have
    \begin{align*}
        \Bigl\| \sum_{x \in \CS(V)} \ket{V_{x,z}}\bra{V_{x,z}}\!\! \sum_{z' \in \CS(W^\perp)} \ket{W_{x',z'}}&\bra{W_{x',z'}}\Bigr \|\\
        &\leq
        \sqrt{2^{\dim(V \cap W) - k}}
    \end{align*}
\end{lemma}
\begin{proof}
    We have $\Pi_{x'+W} := \sum_{z' \in \CS(W^\perp)} \ket{W_{x',z'}}\bra{W_{x',z'}} = \sum_{w \in x'+W} \ket{w}\bra{w}$, which produces the same chain of expressions as in Lemma 3.2 of \cite{Game:culf2022monogamy}, except for the factor of $\frac{1}{2^k}$ at the end:
    \begin{align*}
         &\Bigl\| \sum_{x \in \CS(V)} \ket{V_{x,z}}\bra{V_{x,z}}\!\! \sum_{z' \in \CS(W^\perp)} \ket{W_{x',z'}}\bra{W_{x',z'}}\Bigr \|^2\\
         &=\Bigl\| \Pi_{x'+W} \Bigl (\sum_{x \in \CS(V)} \ket{V_{x,z}}\bra{V_{x,z}} \Bigr ) \Pi_{x'+W}\Bigr \|\\
         & \leq \max_{x \in CS(V)} \Bigl\| \Pi_{x'+W} \ket{V_{x,z}}\bra{V_{x,z}} \Pi_{x'+W}\Bigr \|\\
         &= \max_{x \in CS(V)} \bra{V_{x,z}} \Pi_{x'+W} \ket{V_{x,z}}\\
         &= \max_{x \in CS(V)} \frac{|(x+V)\cap(x'+W)|}{2^k} = 2^{\dim(V \cap W) - k}
    \end{align*}
\end{proof}

The next two lemmas concern intersections of $k$-dimensional subspaces of $\F_2^n$. Lemma \ref{lem:subspaces} is derived from results of finite geometry that can be found, for instance, in \cite{Subspace:braun2018q}:

\begin{lemma}\label{lem:subspaces}
    For a subspace $W \in \Gr_2(n,k)$ and $0\leq m \leq k$, there are $2^{(k-m)^2}\binom{n-k}{k-m}_2\binom{k}{m}_2$ subspaces $V \in \Gr_2(n,k)$ with $\dim(V \cap W) = m$.
\end{lemma}

\begin{proof}
    A special case of Lemma 1 from \cite{Subspace:braun2018q} states that for any subspace $U \subseteq W$ of dimension $m$, there are $2^{(k-m)^2}\binom{n-k}{k-m}_2$ subspaces $V \in \Gr_2(n,k)$ such that $V\cap W = U$. There are $\binom{k}{m}_2$ subspaces $U \subseteq W$ with dimension $m$, which completes the proof.
\end{proof}

We say that a mapping $\pi : \Gr_2(n,k) \to \Gr_2(n,k)$ has the $m$-intersection property if $\dim(W \cap \pi(W)) = m$ for every subspace $W \in \Gr_2(n,k)$. Lemma \ref{lem:subspacePermutations} asserts the existence of a collection of mutually orthogonal permutations with the $m$-intersection property, and follows from Lemma \ref{lem:subspaces}:

\begin{lemma}\label{lem:subspacePermutations}
    For $0 \leq m \leq k$, there exists a collection of $2^{(k-m)^2}\binom{n-k}{k-m}_2\binom{k}{m}_2$ mutually orthogonal permutations $\{\pi_i\}_i$ on $\Gr_2(n,k)$ with the $m$-intersection property.
\end{lemma}

\begin{proof}
    Let $f(n,k,m) = 2^{(k-m)^2}\binom{n-k}{k-m}_2\binom{k}{m}_2$. If $m=k$, the identity permutation is sufficient. Otherwise, let $G_m$ denote the graph with vertex set $V_{G_m} = \Gr_2(n,k)$, and an edge between two vertices $V,W$ iff $\dim(V\cap W)=m$. By lemma \ref{lem:subspaces}, $G_m$ is a $f(n,k,m)-$regular graph, where $f(n,k,m)$ is even. By Petersen's 2-factor theorem (see \cite{2FT:lovasz2009matching}), the edges of $G_m$ can be partitioned into $\frac{1}{2}f(n,k,m)$ edge-disjoint 2-factors.

    Each 2-factor is a collection $\mathcal{C}$ of cycles of length $\geq 3$ such that every vertex belongs to exactly one cycle. We define 2 permutations for each 2-factor: for each cycle $C$ of a 2-factor, fix an orientation of $C$ and define $\pi_1(W) = V$ for each $W \in C$ such that $(W,V)$ is a directed edge in $C$. Since each vertex belongs to exactly one cycle, then $\pi_1$ is a well-defined permutation on $\Gr_2(n,k)$. We define $\pi_2 = \pi_1^{-1}$, and note that since each cycle has length $\geq 3$, then $\pi_1(W) \neq \pi_2(W)$ for every $W\in \Gr_2(n,k)$, hence $\pi_1$ and $\pi_2$ are mutually orthogonal.

    Since each 2-factor is edge-disjoint, two permutations from different 2-factors must necessarily be mutually orthogonal, as for any permutation $\pi$ defined from a 2-factor, $\pi(W)=V$ coincides with an edge between $V$ and $W$. Hence, we have constructed $f(n,k,m)$ mutually orthogonal permutations $\{\pi_i\}_i$ on $\Gr_2(n,k)$ with the $m$-intersection property.
\end{proof}
Additionally, a permutation $\pi_1$ on $\Gr_2(n,k)$ with the $m_1$-intersection property is mutually orthogonal to a permutation $\pi_2$ on $\Gr_2(n,k)$ with the $m_2$-intersection property when $m_1 \neq m_2$. Thus, a corollary of Lemma \ref{lem:subspacePermutations} states that there is a collection of $\binom{n}{k}_2 = \sum_{m=0}^k 2^{(k-m)^2}\binom{n-k}{k-m}_2\binom{k}{m}_2$ mutually orthogonal permutations on $\Gr_2(n,k)$, such that $2^{(k-m)^2}\binom{n-k}{k-m}_2\binom{k}{m}_2$ have the $m$-intersection property for $0 \leq m \leq k$.


\section{Game Definition}\label{sec:def}
\begin{figure}[t]
    \centering
    \includegraphics[width=\linewidth]{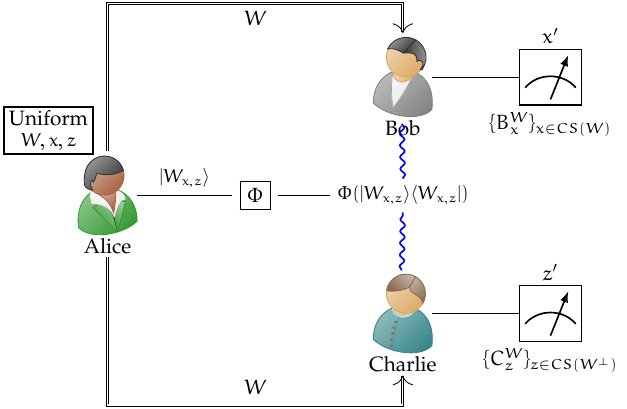}
    \caption{A representation of the game dynamics for the $(n,k)$ Coset Monogamy Game. Alice chooses parameters $W \in \Gr_2(n,k)$, $x,z \in \F_2^n$ uniformly at random and prepares the coset state $\ket{W_{x,z}}$. She sends the coset state over the quantum channel $\Phi$ to Bob and Charlie before publicly broadcasting the subspace choice $W \in \Gr_2(n,k)$. Bob and Charlie measure their subsystems with POVMs to determine guesses for $x$ and $z$ up to cosets of $W$ and $W^\perp$, respectively.}
    \label{fig:game}
\end{figure}

\begin{algorithm}[t]
\caption{\textbf{Protocol 1:} $(n,k)$ Coset Monogamy Game}
\begin{algorithmic}
\State \textbf{Paramaters:}
\begin{itemize}
    \item Finite-dimensional Hilbert spaces $\mcH_B$, $\mcH_C$.
    \item A quantum channel $\Phi: (\C^2)^{\otimes n} \to \mcH_B \otimes \mcH_C$.
    \item A collection $\mathcal{B}$ of of POVMs $\{B_x^W\}_{x \in \CS(W)}$ on $\mcH_B$ for every subspace $W \in \Gr_2(n,k)$.
    \item A collection $\mathcal{C}$ of of POVMs $\{C_z^W\}_{z \in \CS(W^\perp)}$ on $\mcH_C$ for every subspace $W \in \Gr_2(n,k)$.
\end{itemize}
\State \textbf{Game Procedure:}
\begin{itemize}
    \item Alice samples $W \in \Gr_2(n,k)$ and strings $x,z \in \F_2^n$ uniformly at random, and prepares the state $\ket{W_{x,z}}$.
    \item Alice sends the state to Bob and Charlie using the channel $\Phi$. Together, Bob and Charlie receive the state $\Phi(\ket{W_{x,z}}\bra{W_{x,z}}) \in \mcH_B \otimes \mcH_C$.
    \item Alice announces $W \in \Gr_2(n,k)$ to Bob and Charlie.
    \item Bob and Charlie measure their (possibly entangled) states using the POVMs $\{B_x^W\}_{x \in \CS(W)}$ and $\{C_z^W\}_{z \in \CS(W^{\perp})}$ respectively to receive measurement outcomes $x'$ and $z'$.
\end{itemize}
\State \textbf{Win Condition:} $x-x' \in W$ and $z-z' \in W^{\perp}$ 
\end{algorithmic}
\label{alg:game}
\end{algorithm}

The $(n,k)$ Coset Monogamy Game procedure is formally outlined in Protocol~\ref{alg:game}, and illustrated in Fig.~\ref{fig:game}. The game requires Bob and Charlie to extract complementary information $(x,z)$ about a randomly chosen coset state $\ket{W_{x,z}}$ prepared by Alice, conditioned on knowledge of $W \in \Gr_2(n,k)$. Alice first sends the coset state over a channel $\Phi$ to Bob and Charlie, where $\Phi$ can be any quantum channel Bob and Charlie choose. Subsequently, Alice announces $W$, after which Bob and Charlie must devise a measurement to extract information about $x$ and $z$, respectively, without communicating. If Bob and Charlie determine $x$ and $z$ up to a coset of $W$ and $W^\perp$ respectively, then they win the $(n,k)$ Coset Monogamy Game.

No restriction is imposed on the sizes of the Hilbert spaces of Bob and Charlie (i.e., the number of qubits available to Bob and Charlie), so they may use ancillary qubits if desired. The information $(x,z)$ that Bob and Charlie have to extract will be distributed across both Hilbert spaces $\mcH_B$ and $\mcH_C$ randomly based on the choice of $W \in \Gr_2(n,k)$. The fact that Bob and Charlie cannot always simultaneously win the $(n,k)$ Coset Monogamy Game without communicating relates to two fundamental principles of quantum information science: the no-cloning theorem and the monogamy-of-entanglement. The no-cloning theorem \cite{NoCloning:wootters1982single, NoCloning:dieks1982communication} negates the existence of a cloning map $\Phi_{cl}(\ket{\psi}\bra{\psi}) = \ket{\psi}\bra{\psi} \otimes \ket{\psi}\bra{\psi} \in \mcH_B \otimes \mcH_C$ that would allow both Bob and Charlie to discriminate any set of orthogonal states perfectly and always win. For monogamy-of-entanglement, the Choi-Jamio\l{}kowski isomorphism \cite{CJIso:choi1975completely,CJIso:jamiolkowski1972linear} relates the quantum channel $\Phi$ to a tripartite state $\rho_{ABC}$. The stronger the entanglement of the state $\rho_{AB}$ $(\rho_{AC})$ between Alice and Bob (Charlie), the better Bob (Charlie) will be able to discriminate states of a random basis in the $(n,k)$ Coset Monogamy Game. Informally, the principle of monogamy-of-entanglement describes a trade-off between the amount of entanglement in $\rho_{AB}$ versus $\rho_{AC}$ of a tripartite state $\rho_{ABC}$. For a more formal discussion of the principle of monogamy-of-entanglement, refer to \cite{MOE:coffman2000distributed, MOE:osborne2006general}.

From Protocol 1, the winning probability $p_{win}$ of the Coset Monogamy Game can be expressed as
\begin{equation}\label{eq:2PGame}
    p_{win} = \underset{\substack{W \in \Gr_2(n,k) \\ x \in \CS(W) \\ z \in \CS(W^\perp)}}{\E} \Tr[(B_x^W \otimes C_z^W)\Phi(\ket{W_{x,z}}\bra{W_{x,z}})]
\end{equation}
where $\E$ denotes a uniform expectation. We denote the optimal winning probability for the $(n,k)$ Coset Monogamy Game by $p_{k,n} = \sup p_{win}$, where the supremum is taken over all the parameters in Protocol 1. For a given subspace rate $R \in [0,1]$, the optimal winning rate $\omega(R)$ is given by $\limsup_n \exp(\frac{1}{n}\log(p_{\left \lfloor nR \right \rfloor, n}))$.


\section{An Upper Bound on the Winning Rate}\label{sec:thm}

In \cite{Game:culf2022monogamy}, the authors prove that $\omega(\frac{1}{2}) \leq \cos(\frac{\pi}{8})$. In this paper, we strengthen this result and provide a generalization for arbitrary subspace rates $R \in [0,1]$. This proof will proceed very similarly to the proof of Theorem 2.1 in \cite{Game:culf2022monogamy}. Still, we strengthen the upper bound from \cite{Game:culf2022monogamy} by applying Lemma \ref{lem:NormOfSum} directly on the sum of $\binom{n}{k}_2$ projections seen in \eqref{eq:SumOfPi}.

\begin{theorem}\label{thm:mainTheorem}
    $\omega(R) \leq 2^{-\frac{1}{2}R^*}$, where $R^* = \min \{R,1-R\}$
\end{theorem}

\begin{proof}
Fix $0 \leq k \leq n$ and a strategy of parameters for Protocol 1. As in \cite[Lemma 3.1]{Game:culf2022monogamy}, we use the Choi-Jamio\l{}kowski isomorphism to transform \eqref{eq:2PGame} into the expression of an extended non-local game \cite{NLG:johnston2016extended}:

\begin{equation}\label{eq:NLG}
    p_{win} = \underset{W}{\E} \sum_{x,z} \Tr[(\ket{W_{x,z}}\bra{W_{x,z}} \otimes B_x^W \otimes C_z^W)\rho]
\end{equation}

where the sum is over $x \in \CS(W)$, $z \in \CS(W^\perp)$, the expectation is over $W \in \Gr_2(n,k)$, and $\rho = (\Id_A \otimes \Phi)(\ket{\phi}\bra{\phi})$ for any maximally entangled state $\ket{\phi} \in \mcH_A \otimes \mcH_A$. Furthermore, we use a standard purification argument and Naimark's dilation theorem as in \cite{Game:culf2022monogamy,Game:tomamichel2013monogamy} to assume without loss of generality that $\rho$ is a pure state $\ket{\psi}\bra{\psi}$, and all POVMs $\{B_x^W\}_x$, $\{C_z^W\}_z$ are PVMs.\\

Denoting $\Pi^W = \sum_{x,z}\ket{W_{x,z}}\bra{W_{x,z}} \otimes B_x^W \otimes C_z^W$, we upper bound the winning probability as 
\begin{equation}\label{eq:sumNorm}
    p_{win} \leq \frac{1}{N} \Bigl \| \sum_{W \in \Gr_2(n,k)} \Pi^W \Bigr \|
\end{equation}
where $N = \binom{n}{k}_2$. We briefly remark here for completeness that equality in  \eqref{eq:sumNorm} is achievable, even if the principal eigenvectors $\ket{\psi}$ of $\sum_W \Pi^W$ cannot be obtained from the 
Choi-Jamio\l{}kowski isomorphism. The argument involves enlarging Bob and Charlie's Hilbert spaces, and purifying a principal eigenvector $\ket{\psi}$ of $\sum_W \Pi^W$ to a state $\ket{\psi'}$ that is maximally entangled between Alice's Hilbert space and the joint Hilbert space of Bob and Charlie. The state $\ket{\psi'}$ can then be obtained from the Choi-Jamio\l{}kowski isomorphism by an appropriate quantum channel $\Phi$.

We next derive an upper bound $\bigl \| \Pi^V \Pi^{W} \bigr \|$. Defining projections $
P^V = \sum_{x,z} \ket{V_{x,z}}\bra{V_{x,z}} \otimes \Id_B \otimes C_z^V$ and 
$Q^W = \sum_{x,z} \ket{W_{x,z}}\bra{W_{x,z}} \otimes B_x^W \otimes \Id_C
$, we have
\begin{equation}
\Bigl \| \Pi^V \Pi^{W} \Bigr \| \leq \Bigl \| P^V Q^W \Bigr \|
\label{eq:PQbound}
\end{equation}
which holds since $0 \leq \Pi^V \leq P^V$ and $0 \leq \Pi^W \leq Q^W$ (see, for example, \cite[Lemma 1]{Game:tomamichel2013monogamy}). An explicit calculation of $\Bigl \| P^V Q^W \Bigr \|$ gives the same upper bound as in Lemma \ref{lem:prodOfSumOfCosetStates}:
\begin{align*}
    \Bigl \| P^V &Q^W \Bigr \|\\
    &= \Bigl \| \sum_{\substack{x,z\\x',z'}} \ket{V_{x,z}}\bra{V_{x,z}}\ket{W_{x',z'}}\bra{W_{x',z'}} \otimes B_{x'}^W \otimes C_z^V \Bigr \|\\
    &= \max_{x',z} \Bigl \| \sum_{x,z'} \ket{V_{x,z}}\bra{V_{x,z}}\ket{W_{x',z'}}\bra{W_{x',z'}} \Bigr \|\\
    &\leq \sqrt{2^{\dim(V \cap W) - k}}
\end{align*}
which gives $ \Bigl \| \Pi^V \Pi^W \Bigr \| \leq \sqrt{2^{\dim(V \cap W) - k}}$. The second equality occurs since $P^VQ^W$ is the sum over $z \in \CS(V),x' \in \CS(W)$ of operators with orthogonal domain and range.

From Lemma \ref{lem:subspacePermutations}, we obtain $N$ mutually orthogonal permutations on $\Gr_2(n,k)$ such that $2^{(k-m)^2}\binom{n-k}{k-m}_2\binom{k}{m}_2$ have the $m$-intersection property for $0 \leq m \leq k$. Applying Lemma \ref{lem:NormOfSum} with these permutations upper bounds the norm of $\sum_W \Pi^W$ in terms of the norms of products of projectors $\Pi^W$:
\begin{equation}\label{eq:SumOfPi}
    \Bigl \| \sum_{W \in \Gr_2(n,k)} \Pi^W \Bigr \| \leq \sum_{i=1}^N \max_{W \in \Gr_2(n,k)} \Bigl \| \Pi^W \Pi^{\pi_i(W)} \Bigr \|
\end{equation}
(Each $\Pi^W$ is a projection, and we omitted the square roots.)
\\

Now, recall that 1) $\Bigl \| \Pi^W \Pi^{\pi_i(W)} \Bigr \| \leq \sqrt{2^{\dim(W \cap \pi_i(W))}}$, and 2) $2^{(k-m)^2}\binom{n-k}{k-m}_2\binom{k}{m}_2$ of the chosen permutations have the $m$-intersection property. We can then upper bound the right-hand side of \eqref{eq:SumOfPi} as $\sum_{m=0}^k 2^{(k-m)^2}\binom{n-k}{k-m}_2\binom{k}{m}_2 \sqrt{2^{m-k}}$, yielding the following upper bound for the probability of winning the $(n,k)$ Coset Monogamy Game:
\begin{align*}
      p_{win} & \leq \frac{1}{N}\sum_{m=0}^k 2^{(k-m)^2}\binom{n-k}{k-m}_2\binom{k}{m}_2 \sqrt{2^{m-k}}\\
      &= \frac{1}{N}\sum_{m=0}^k 2^{m^2}\binom{n-k}{m}_2\binom{k}{m}_2 \sqrt{2^{-m}}\\
      & = O(\sqrt{2^{-\min\{k,n-k\}}})
\end{align*}
where we used the reflexive relation $\binom{k}{k-m}_2=\binom{k}{m}_2$ for Gaussian binomial coefficients. To prove that the summation $g(n,k)=\frac{1}{N}\sum_{m=0}^k 2^{m^2}\binom{n-k}{m}_2\binom{k}{m}_2 \sqrt{2^{-m}}$ satisfies $O(\sqrt{2^{-\min\{k,n-k\}}})$, we first observe that when $k \geq \frac{n}{2}$, the contribution of the terms for $m > n-k$ is zero. Hence, we have the equivalence $g(n,k)=g(n,n-k)$, and we may replace $k$ with $n-k$ to assume without loss of generality that $k \leq \frac{n}{2}$.

Recalling $f(n,k,m)$ as in Lemma \ref{lem:subspacePermutations}, we observe the following for $0 \leq m \leq k-2$:

\begin{align*}
    \frac{f(n,k,k-m)}{f(n,k,k-m-1)} = \frac{2^{-(2m+1)}(2^{m+1}-1)^2}{(2^{n-k-m}-1)(2^{k-m}-1)} \leq \frac{2}{9}
\end{align*}

Thus, $f(n,k,k-m)$ is an increasing function of $m$ for $0\leq m \leq k-1$. An inductive argument yields $f(n,k,k-m)\leq \left( \frac{2}{9} \right)^{k-1-m}f(n,k,1)$ for $0 \leq m \leq k-1$, which allows us to upper bound $g(n,k)$ as follows:

\begin{align*}
   & g(n,k) = \frac{1}{N}\sum_{m=0}^k 2^{m^2}\binom{n-k}{m}_2\binom{k}{m}_2\sqrt{2^{-m}}\\
    &= \frac{1}{N} \sum_{m=0}^{k} f(n,k,k-m)\sqrt{2^{-m}}\\
    &\leq \frac{f(n,k,1)}{N} \left(\frac{2}{9}\right)^{k-1}\sum_{m=0}^{k-1} \left(\frac{9}{2\sqrt{2}}\right)^m + \frac{f(n,k,0)}{N}\sqrt{2^{-k}}\\
    &\leq \left(\frac{2}{9}\right)^{k-1}\left(\frac{\left(\frac{9}{2\sqrt{2}}\right)^k - 1}{\frac{9}{2\sqrt{2}}-1}\right) + \sqrt{2^{-k}}\\
    & \leq \left(\frac{9}{2(\frac{9}{2\sqrt{2}}-1)} + 1\right) \sqrt{2^{-k}} = O(2^{-\frac{1}{2}k})
\end{align*}

Therefore, for $R \in [0,\frac{1}{2}]$, we have $p_{\lfloor nR \rfloor,n}=O(2^{-\frac{1}{2}\lfloor nR \rfloor})$, giving $\omega(R) \leq 2^{-\frac{1}{2}R}$. For $R \in (\frac{1}{2},0]$, we have $p_{\lfloor nR \rfloor,n}=O(2^{-\frac{1}{2}(n - \lfloor nR \rfloor )})$ for sufficiently large $n$, giving $\omega(R) \leq 2^{-\frac{1}{2}(1-R)}$, proving Theorem \ref{thm:mainTheorem}.
\end{proof}
The reflexive nature of Theorem \ref{thm:mainTheorem} can be explained by the dual relation of coset states. Namely, given a strategy $(\mcH_B,\mcH_C,\mathcal{B},\mathcal{C},\Phi)$ for a $(n,k)$ Coset Monogamy Game, there exists a strategy $(\mcH_{B'},\mcH_{C'},\mathcal{B}',\mathcal{C}',\Phi')$ for a $(n,n-k)$ Coset Monogamy Game that achieves the same winning probability by setting $\mcH_{B'} = \mcH_C$, $\mcH_{C'} = \mcH_B$, $B_z^{W^\perp} = C_z^W$ for $z \in \CS(W^\perp)$, $C_x^{W^\perp} = B_x^W$ for $x \in \CS(W)$, and $\Phi'(X) = \SWAP \left( \Phi(H^{\otimes n} X H^{\otimes n}) \right)\SWAP^\dagger$. Here, $\SWAP: \mcH_B \otimes \mcH_C \to \mcH_C \otimes \mcH_B$ denotes the unitary map sending $\ket{\psi} \otimes \ket{\phi}$ to $\ket{\phi} \otimes \ket{\psi}$.


\section{Unentangled Extended Game Strategies}\label{sec:other}
Extended non-local games are a generalization of non-local games, in which a multipartite state $\rho$ is shared among a referee and all players, and the referee (in our case, Alice) measures according to some hermitian observable based upon the posed questions and received answers \cite{NLG:johnston2016extended}. These games provide a natural setting for formulating multipartite steering inequalities \cite{Steering:uola2020quantum} and monogamy-of-entanglement inequalities \cite{Game:tomamichel2013monogamy}. 

In the extended non-local formulation of the Coset Monogamy Game, an unentangled strategy involves restricting the choice of the tripartite state $\rho$ in \eqref{eq:NLG} to the set of separable states $\rho = \sum_{i=1}^N p_i \rho_A^i \otimes \rho_{BC}^i$. The winning probability of the unentangled strategies is in an upper bound of the deterministic strategies, where Bob and Charlie choose $x$ and $z$, respectively, as a deterministic function of $W$ \cite{NLG:johnston2016extended}. Thus, the optimal winning probability $p_{opt}^{u}(n,k)$ for unentangled strategies is given by
\begin{equation}
    p_{opt}^{u}(n,k) = \max_{f,g} \Bigl \| \underset{W}{\E} \ket{W_{f(W),g(W)}} \bra{W_{f(W),g(W)}} \Bigr \|
\end{equation}
where $f,g:\Gr_2(n,k) \to \F_2^n$, and the expectation is uniform over $W \in \Gr_2(n,k)$. For unentangled strategies, we can explicitly determine $p_{opt}^u(n,k)$.

\begin{theorem}\label{thm:unentangled}
    $p_{opt}^{u}(n,k) = \frac{1}{N}\sum_{m=0}^k 2^{m^2} \binom{n-k}{m}_2\binom{k}{m}_22^{-m} = O(2^{-\min\{k,n-k\}})$
\end{theorem}

\begin{proof}
    The value is attained by choosing $f(W)=g(W)=0$ for all $W\in \Gr_2(n,k)$ and the eigenvector $\ket{\psi} \propto \sum_{W \in \Gr_2(n,k)}\ket{W}$. To show this value is optimal, we fix deterministic strategies $f,g$ and use Lemma \ref{lem:NormOfSum} with the same permutations as in Theorem \ref{thm:mainTheorem}. Denoting $\Pi_{f,g}^W = \ket{W_{f(W),g(W)}} \bra{W_{f(W),g(W)}}$ gives
    \begin{align*}
         &\Bigl \| \underset{W \in \Gr_2(n,k)}{\E} \Pi_{f,g}^W \Bigr \| \leq \frac{1}{N} \sum_{i=1}^N \max_{W \in \Gr_2(n,k)} \Bigl \| \Pi_{f,g}^W \Pi_{f,g}^{\pi_i(W)}  \Bigr \|\\
         &= \frac{1}{N} \sum_{i=1}^N \max_{W \in \Gr_2(n,k)} \left|\bra{W_{f(W),g(W)}}\ket{\pi_i(W)_{f(W),g(W)}}\right|\\
         &\leq \frac{1}{N} \sum_{m=0}^k 2^{(k-m)^2}\binom{n-k}{k-m}_2\binom{k}{m}_22^{m-k}\\
         &= \frac{1}{N} \sum_{m=0}^k 2^{m^2}\binom{n-k}{m}_2\binom{k}{m}_22^{-m}
    \end{align*}
    where we used \eqref{eq:innerProd} to determine the absolute value of the inner products. The fact that the summation of the final line is $O(2^{-\min\{k,n-k\}})$ follows from the exact same reasoning for the summation $g(n,k)$ in Theorem \ref{thm:mainTheorem}.
\end{proof}

\section{Conclusion}
In this work, we extended the Coset Monogamy Game of \cite{QMoney:coladangelo2021hidden,Game:culf2022monogamy} to a more general setting, which we refer to as the $(n,k)$ Coset Monogamy Game, in which the players Bob and Charlie must extract complementary information of unequal sizes from a random coset state of dimension $k=nR$. Our primary interest was classifying the optimal winning rate $\omega(R)$ as a function of the underlying subspace rate $R \in [0,1]$ as $n$ grows arbitrarily large. By using operator inequalities in a manner similar to \cite{Game:tomamichel2013monogamy}, we deduce a non-trivial convex upper bound for $\omega(R)$ for every $R \in (0,1)$. In particular, we show that for every $R \in (0,1)$, the optimal winning probability for the $(n,nR)$ Coset Monogamy Game decays exponentially in the parameter $n$. Additionally, for the weaker class of unentangled strategies, we characterized the exact optimal winning probability and demonstrated achievability using a particular strategy.

The key to our results was Lemma \ref{lem:subspacePermutations}, which, to the authors' knowledge, is a novel classical result regarding permutations on $\Gr_2(n,k)$ with the property that the intersection of any subspace and its image has a constant dimension. The Lemma is used to optimally pair subspaces $V,W$ whose projectors $\Pi^V, \Pi^W$ have minimal overlap in \eqref{eq:SumOfPi}. In fact, Lemma \ref{lem:subspacePermutations} gives the strongest possible result on the optimal winning rate using the techniques of \cite{Game:tomamichel2013monogamy,Game:culf2022monogamy} to bound the optimal winning probability (Lemmas \ref{lem:NormOfSum} and \ref{lem:prodOfSumOfCosetStates}). This is true because for subspaces $V,W \in \Gr_2(n,k)$, the smallest obtainable upper bound on the overlaps of $\Pi^V$, $\Pi^W$ using Lemma \ref{lem:prodOfSumOfCosetStates} is $\sqrt{2^{-\min \{k,n-k\}}}=2^{-\frac{1}{2}nR^*}$, which is exactly the obtained upper bound for $\omega(R)^n$.

Combining the results of Sections \ref{sec:thm}, \ref{sec:other} and specializing to the case $R=\frac{1}{2}$, we strengthen the results from \cite{QMoney:coladangelo2021hidden,Game:culf2022monogamy},  showing that the optimal winning probability $p_n$ for the Coset Monogamy Game of \cite{QMoney:coladangelo2021hidden,Game:culf2022monogamy} satisfies $2^{-\frac{1}{2}n} \leq p_n \leq 2^{-\frac{1}{4}n}$. However, as pointed out in \cite{BlackHole:Poremba2024}, the Coset Monogamy Game is not strong enough as an uncloneable cryptographic primitive (specifically, a cloning game \cite{ananth2023cloning}) to give adversaries a trivial winning probability of $O(2^{-n})$.

\section*{Acknowledgements}
Special thanks to Vincent Russo for spotting errors in earlier versions of this paper.

\bibliography{Allerton2025}
\bibliographystyle{IEEEtran}
\end{document}